\documentclass[letterpaper, 10 pt, conference]{ieeeconf}

\newcommand{\mbb}{\mathbb}
\newcommand{\mbf}{\mathbf}
\newcommand{\mc}{\mathcal}

\newcommand{\termfont}[1]{{\tt #1}}

\newcommand{\ceil}[1]{\lceil {#1} \rceil}

\newcommand{\state}{x}
\newcommand{\goalstate}{\hat \state}

\newcommand{\control}{u}

\newcommand{\bcontrol}{\mbf{\control}}

\newcommand{\dyn}{f}
\newcommand{\fdyn}{\dyn^+}
\newcommand{\bdyn}{\dyn^-}
\newcommand{\cdyn}{F}

\newcommand{\runningcost}{g}

\newcommand{\algname}{GrAVITree}

\newcommand{\xdim}{n}
\newcommand{\udim}{m}
\newcommand{\xconstraint}{\mc{X}}
\newcommand{\uconstraint}{\mc{U}}
\newcommand{\xset}{\mbb{R}^\xdim}
\newcommand{\uset}{\mbb{R}^{\udim}}

\newcommand{\valuefn}{J}

\newtheorem{remark}{Remark}

\newtheorem{proposition}{Proposition}

\newcommand{\figref}[1]{Fig.~\ref{#1}}
\newcommand{\secref}[1]{Sec.~\ref{#1}}
\usepackage{graphicx}
\usepackage{adjustbox}
\usepackage{wrapfig}
\usepackage[labelfont=bf,figurename=Fig.,font=small]{caption}
\usepackage{subcaption}
\usepackage{lipsum}
\usepackage{amssymb}
\usepackage{amsmath}
\usepackage{siunitx}
\usepackage{mathtools}
\usepackage{xcolor}
\usepackage{hyperref}
\usepackage[ruled,vlined,linesnumbered]{algorithm2e}
\usepackage[noend]{algpseudocode}
\usepackage{ifthen}
\usepackage{dsfont}
\usepackage[backend=biber,style=numeric-comp,sorting=none]{biblatex}
\usepackage{balance}

\bibliography{books, papers}

\IEEEoverridecommandlockouts
\title{\LARGE \bf GrAVITree: Graph-based Approximate Value Function In a Tree}

\author{
Patrick H. Washington, David Fridovich-Keil, and Mac Schwager
\thanks{
P. Washington and M. Schwager are with the Department of Aeronautics \& Astronautics, Stanford University. D. Fridovich-Keil is with the Department of Aerospace Engineering, UT Austin. Correspondence to \href{mailto:phw@stanford.edu}{\termfont{phw@stanford.edu}}.}%
\thanks{Toyota Research Institute provided funds to support this work. The NASA University Leadership initiative (grant \#80NSSC20M0163) provided funds to assist the authors with their research, but this article solely reflects the opinions and conclusions of its authors and not any NASA entity. The first author was supported on a National Defense Science and Engineering Graduate (NDSEG) Fellowship. We are grateful for this support.}}

\begin{document}

\maketitle
\thispagestyle{empty}
\pagestyle{empty}

\begin{abstract}
In this paper, we introduce \algname, a tree- and sampling-based algorithm to compute a near-optimal value function and corresponding feedback policy for indefinite time-horizon, terminal state-constrained nonlinear optimal control problems. 
Our algorithm is suitable for arbitrary nonlinear control systems with both state and input constraints. 
The algorithm works by sampling feasible control inputs and branching backwards in time from the terminal state to build the tree, thereby associating each vertex in the tree with a feasible control sequence to reach the terminal state. 
Additionally, we embed this stochastic tree within a larger graph structure, rewiring of which enables rapid adaptation to changes in problem structure due to, e.g., newly detected obstacles. 
Because our method reasons about global problem structure without relying on (potentially imprecise) derivative information, it is particularly well suited to controlling a system based on an imperfect deep neural network model of its dynamics. 
We demonstrate this capability in the context of an inverted pendulum, where we use a learned model of the pendulum with actuator limits and achieve robust stabilization in settings where competing graph-based and derivative-based techniques fail.
\end{abstract}

\section{Introduction}
\label{sec:intro}

It is generally difficult to find global solutions to nonlinear optimal control problems, and it can be particularly difficult to do so in the presence of actuation and state constraints. 
Most existing approaches rely on either approximate solutions to the Hamilton-Jacobi-Bellman (HJB) equation or settle for locally-convergent model predictive control (MPC) techniques.
HJB and path integral methods provide an implicit state feedback policy; however, they are computationally prohibitive in the general case. 
Indeed, HJB methods are often only effective when they can be computed offline, and are hence not well suited to settings in which problem structure changes at run-time, e.g. due to a new obstacle.
Conversely, nonlinear MPC methods can react online to such changes, but do not offer a feedback policy, and can suffer from impractically long optimization times for nonlinear dynamics, non-convex constraints and high state dimensions.  
Furthermore, nonlinear MPC cannot give optimality guarantees for solution trajectories in such non-convex cases. Most importantly, solvers for nonlinear MPC use Jacobians, and sometimes Hessians, of the dynamics model \cite{pfrommer2022tasil}, which can be quite unreliable in deep learned dynamics models, leading to poor control performance.

We propose a sampling-based solution technique, called \textbf{Gr}aph-based
\textbf{A}pproximate \textbf{V}alue Function \textbf{I}n a \textbf{Tree} (\algname), to approximately solve the HJB equation for a class of constrained nonlinear optimal control problems. Our method avoids finite element grids common to other HJB solution methods, and, unlike other HJB methods, maintains a graph of multiple possible paths to a goal state to enable fast online ``rewiring'' in response to changing constraint structure. Furthermore, our method does not rely on gradients of the dynamics model, making it suitable for controlling systems whose dynamics are modeled by a deep neural network.

As illustrated in \figref{fig:front}, we construct a tree incrementally \emph{backward in time} from a desired terminal state and rewire the tree such that trajectories through the tree are optimal given the discrete set of state samples in the graph. As the samples fill the dynamically-feasible submanifold of the state space, we observe that optimal trajectories through the tree approach globally optimal trajectories in the continuous state space. Critically, at runtime, we also provide a feedback control law that uses this tree and the current system state to interpolate an input at \emph{any} state, even one which was not sampled in the tree. 

\begin{figure}[t]
    \centering
    \includegraphics[width=0.33\textwidth]{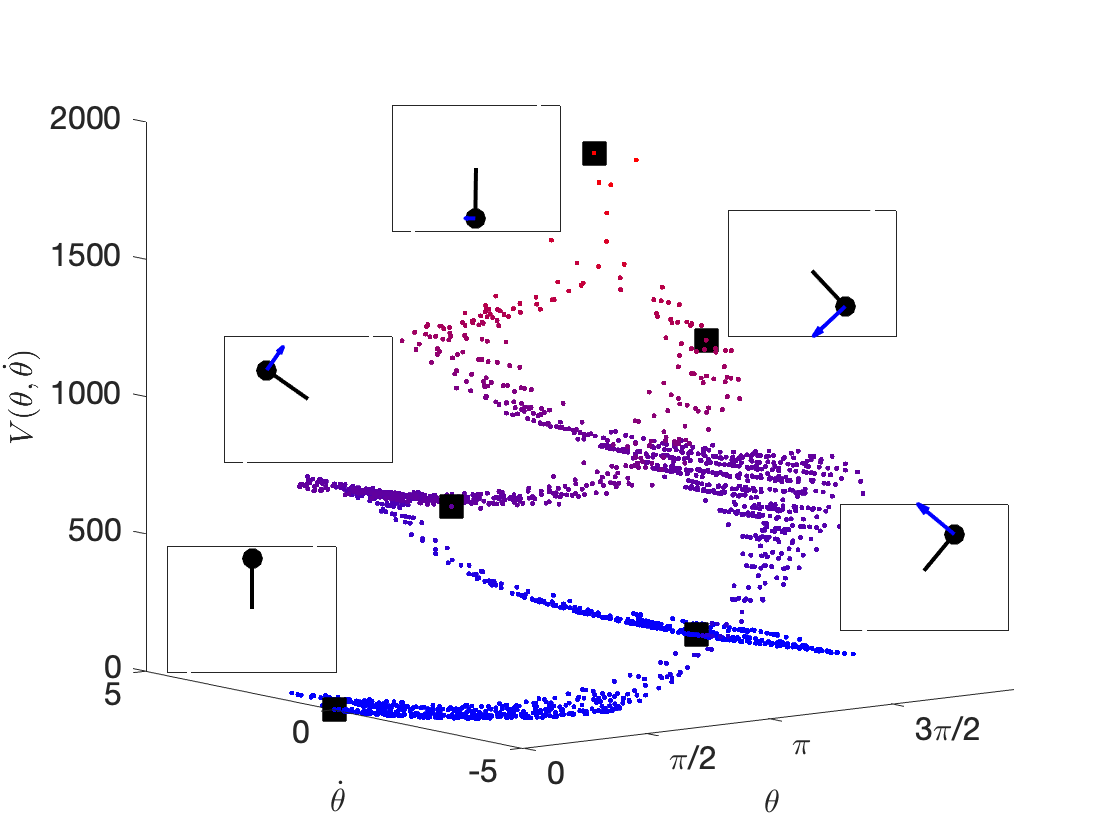}
    \caption{An example of a tree corresponding to pendulum swingup. The tree is constructed backward-in-time.
    Insets show position and speed at illustrative states.}
    \label{fig:front}
\end{figure}

\section{Problem Formulation}
\label{sec:problem}

\subsection{Optimal control problem}
\label{subsec:optimal_control_problem}

Consider the following optimal control problem which is terminal state constrained and has unspecified final time:
\begin{align}
    \label{eqn:terminal_constraint_problem}
    \valuefn(\state_0; \hat x) = \min_{\bcontrol, t' \ge 0}~&\sum_{t = 0}^{t'} \runningcost(\state_t, \control_t) \\
    \label{eqn:input_constraint}
    \textnormal{subject to}~&\control_t \in \uconstraint \subseteq \uset\\
    \label{eqn:dynamics_constraint}
    &\fdyn(\state_t, \control_t) = \state_{t+1} \in \xconstraint \subseteq \xset\\
    \label{eqn:terminal_constraint}
    &\state_{t'} = \hat x.
\end{align}
Here, the system dynamics are given in \eqref{eqn:dynamics_constraint} for state $\state$ and input $\control$. The desired terminal state is $\goalstate$, which must occur at some future time step $t'$ (without loss of generality, we presume that time starts at $t = 0$). Further, we presume that state and input are uniformly bounded within compact sets $\xconstraint$ and $\uconstraint$, respectively.

In \eqref{eqn:terminal_constraint_problem}, the decision variables are $\bcontrol := (\control_0, \dots, \control_{t'})$ and the final time $t' \ge 0$ at which the terminal constraint \eqref{eqn:terminal_constraint} must be satisfied. The objective is then the accumulated running cost $\runningcost$.
Certain problems of this form are well-studied. For example, when $\runningcost(\state, \control) \equiv 0$, \eqref{eqn:terminal_constraint_problem} reduces to a standard reachability problem, for which level-set methods \cite{mitchell2005time, bansal2017hamilton} are available for low-dimensional systems (typically, $\xdim \le 5$). Other methods exist as well, such as ellipsoidal \cite{kurzhanski00ellispoidal, kurzhanski02ellipsoidal} and zonotope \cite{althoff2011zonotope} approximations. However, as we are concerned with the general case in which $\runningcost \not \equiv 0$ and for larger state space dimensions $\xdim$, our approach does not explicitly attempt to either (a) solve a fixed-horizon HJB equation, or (b) build a geometric approximation to any reachable set of interest.

These (and other, still more general) types of problems are commonly solved to local optimality using a variety of nonlinear programming techniques, e.g., \cite{chen2017constrained, nocedal2006numerical}. Unlike these methods, we focus on approximating a global optimum and an associated state-feedback strategy (which is typically beyond the scope of nonlinear programming and MPC-based approaches).

\subsection{Backward dynamics}
\label{subsec:backward_dynamics}

So far, we have presumed that system dynamics are represented in discrete time by \eqref{eqn:dynamics_constraint}. Later, we shall make use of the \emph{backward} form of the dynamics. Before defining the backward dynamics, we first note that if continuous-time dynamics are $\dot x_\tau = \cdyn(x_\tau, u_\tau)$, the discrete time (forward) dynamics $\fdyn$ from \eqref{eqn:dynamics_constraint} can be written as
\begin{align} 
    x_{t+1} = \fdyn(\state_t, \control_t) &= \state_t + \int_{\tau(t)}^{\tau(t+1)} \cdyn\big(\state(\tau), \control(\tau)\big) d\tau
    \label{eqn:forward_dynamics}
\end{align}
where $\tau$ represents the continuous-time value, $\tau(t)$ is the time associated with time step $t$,
and state/input signals are indicated by $\state(\cdot)$ and $\control(\cdot)$, respectively. The form of \eqref{eqn:forward_dynamics} suggests the following equation for backward dynamics $\bdyn$:
\begin{equation} 
    \label{eqn:backward_dynamics}
    x_{t-1} = \bdyn(\state_t, \control_{t-1}) = \state_t + \int_{\tau(t)}^{\tau(t-1)} \cdyn\big(\state(\tau), \control(\tau)\big) d\tau .
\end{equation}

Practically, we will approximate the integrals in \eqref{eqn:forward_dynamics} and \eqref{eqn:backward_dynamics} via numerical quadrature, e.g.,
\begin{equation}
    \label{eqn:approx_integral}
    \int_{\tau_1}^{\tau_2} \cdyn(\state_\tau,\control_\tau) \approx \sum_{k=0}^{\ceil{(\tau_2-\tau_1)/\delta\tau}} F(\state_{\tau_1+k\delta\tau},\control_{\tau_1+k\delta\tau})\delta\tau
\end{equation}
where $\delta\tau \ll 1$ is a small time interval.

Later, in \secref{sec:backward_tree}, we will use the backward form of dynamics \eqref{eqn:backward_dynamics} extensively to construct a tree whose branches encode trajectories which terminate in the desired state $\goalstate$.

\section{Related Work}
\label{sec:related}

In principle, our work is closely related to classical dynamic programming approaches from the optimal control literature. Algorithmically, however, it draws significant inspiration from related work in sample-based kinodynamic motion planning. Both areas are extremely well-studied; what follows is an succinct summary of existing techniques and a discussion of how they relate to the present work.

\subsection{Dynamic programming}
\label{subsec:dynamic_programming}

Problem \eqref{eqn:terminal_constraint_problem} may be viewed as a \emph{shortest-path problem}, and hence admits a dynamic programming solution to compute the optimal path length, or \emph{value function}, $V(\cdot)$ for every state $\state$. This fact forms the basis for a wide variety of techniques; the interested reader is directed to \cite{bertsekas1995dynamic}. Exact dynamic programming methods scale exponentially with the size of discrete state spaces and require discretization to extend to continuous problems. Approximate methods \cite{bertsekas2012approximate} are increasingly popular, however most methods do not come with formal convergence guarantees and computational complexity can still be a challenge with these methods. While they are still quite successful in a number of applications \cite{lee2004approximate}, two areas of particular note are linear-quadratic problems (and approximations) such as \cite{kalman1960contributions, mayne1966second, li2004iterative} and Monte Carlo Tree Search (MCTS) methods, which have been used extensively to solve Partially Observed Markov Decision Processes (POMDPs) \cite{ye2017despot, sunberg2018online}.

Reinforcement learning \cite{sutton2018reinforcement} approaches form a broad class of approximate dynamic programming methods that has received significant attention in recent years. ``Deep'' variants of these methods refine a deep learned representation of the value function (or implicitly represent it with a \emph{policy} or feedback law). Experimental results in this field are extremely encouraging, e.g. in the domains of arcade games \cite{mnih2013playing}, Go \cite{silver2016mastering}, robotic arm control \cite{gu2017deep}, and dexterous manipulation \cite{andrychowicz2020learning}. However, despite these experimental successes and several initial attempts \cite{achiam2017constrained, liu2019ipo, li2021augmented}, it is still uncommon for reinforcement learning methods to consider non-dynamic constraints of the form \eqref{eqn:input_constraint} and \eqref{eqn:terminal_constraint}. Our method accounts for these types of constraints by construction. %

\subsection{Sample-based kinodynamic planning}
\label{subsec:sample_based_kinodynamic_planning}

Sample-based motion planning has also been extraordinarily successful in recent decades. Early approaches such as the well-known rapidly-exploring random tree (RRT) \cite{lavalle1998rapidly} and probabilistic roadmap (PRM) \cite{kavraki1996probabilistic} have since been generalized to work in kinodynamic settings \cite{csucan2009kinodynamic, karaman2013sampling} and guarantee asymptotic optimality \cite{karaman2011sampling, gammell2015batch}. In these algorithms, states are sampled at random from the feasible set $\xconstraint$ and added to a graph structure which encodes the proximity of sampled states. Most methods construct \emph{tree} structures from the current state, the desired terminal state, or both \cite{kuffner2000rrt}. Sampling methods can (a) scale to CPU-constrained \cite{ichnowski2019motion} and high-dimensional systems \cite{luna2020scalable, janson2015fast} and (b) execute in real-time for modest problem instances \cite{allen2016real}. Additionally, sampling-based planning can be coupled with nonlinear stability analysis to provide runtime stability guarantees \cite{tedrake2010lqr}. Although the present work also defines a feedback control law, it does not rely upon solving sum-of-squares programs (and is hence more scalable, but does not guarantee stability). 

Our own work is inspired by the RRT* algorithm of \cite{karaman2011sampling}; in particular, the \termfont{UpdateTree} method (\secref{subsec:update_tree}) is a direct extension of the \termfont{rewire} procedure from RRT* to our setting. Further, RRT* presumes the existence of a \termfont{steer} subroutine to drive a dynamical system to a desired state. In general, this is a difficult problem, and in the same spirit as \cite{lamiraux2004kinodynamic} we contribute an efficient, general method for approximate steering in \secref{subsec:find_connections}. Unlike RRT*, however, we build our tree \emph{backward} from the terminal state since we presume that only it is fixed and the initial state will be determined at runtime. Additionally, we also provide a feedback law that drives arbitrary states, not just those in the tree, to the goal state. Hence, our method is robust to some degree of model mismatch, where the dynamics of the system $\cdyn$ differ from those of the physical system.
\section{\algname}
\label{sec:backward_tree}

We solve problem \eqref{eqn:terminal_constraint_problem} by constructing a graph $\mathcal{G} = (\mathcal{V}_G,\mathcal{E}_G)$ with vertices $\mathcal{V}_G$ and edges $\mathcal{E}_G$, where vertices contain state information and edges contain the control and stage cost information. The graph is directed and may have cycles. 

For some vertex $v_G\in\mathcal{V}_G$, denote the state with $v_G[\state]$.
For each edge $e_G\in\mathcal{E}_G$, denote the control with $e_G[\control]$ and the stage cost with $e_G[c]$.

More precisely, if there is an edge $e_{G}^{ij}$ that connects vertex $v_G^i$ to vertex $v_G^j$, then
\begin{equation}
    v_G^j[\state] = \fdyn(v_G^i[\state],e_G^{ij}[\control]).
\end{equation}
We will construct the tree backward in time; hence,
\begin{equation}
    v_G^i[\state] = \bdyn(v_G^j[\state],e_G^{ij}[\control]).
\end{equation}
Moreover, the stage cost of an edge is
\begin{equation}
    e_G^{ij}[c] = \runningcost(v_G^i[\state],e_G^{ij}[\control]).
\end{equation}

While the graph $\mathcal{G}$ is constructed, we also build a tree $\mathcal{T}=(\mathcal{V}_T, \mathcal{E}_T)$. 
Unlike $\mathcal{G}$, the tree's vertices contain both state and cost-to-go information, and its edges contain controls and the cost of executing those controls from the corresponding states.
Importantly, where $\mathcal{G}$ contains edges from each vertex to all vertices determined to be reachable in one step, in $\mathcal{T}$ each vertex has a single edge corresponding to the best control action to take from that state, according to the current cost-to-go information.

For each vertex $v_T\in\mathcal{V}_T$, denote the state with $v_T[\state]$ and the cost-to-go with $v_T[J]$. For each edge $e_T\in\mathcal{E}_T$, denote the control with $e_T[\control]$ and the stage cost with $e_T[c]$.
As in graph $\mathcal{G}$, tree vertices satisfy
\begin{equation}
    v_T^j[\state] = \fdyn(v_T^i[\state],e_T^{ij}[\control]),~
    v_T^i[\state] = \bdyn(v_T^j[\state],e_T^{ij}[\control])
\end{equation}
and the stage costs encoded in the tree are identical to the those encoded in the graph:
\begin{equation}
    e_T^{ij}[c] = \runningcost(v_T^i[\state],e_T^{ij}[\control]).
\end{equation}

From the time-additive form of \eqref{eqn:terminal_constraint_problem}, the cost-to-go of a vertex is determined by adding the stage cost to the cost-to-go of the parent vertex:
\begin{equation}
    v_T^i[J] = e_T^{ij}[c] + v_T^j[J].
\end{equation}
The root of the tree has $v_T^\text{root}[\state]\equiv\goalstate$ and $v_T^\text{root}[J] \equiv 0$.

Tree vertices are in one-to-one correspondence with graph vertices. Meanwhile, every edge in $\mathcal{T}$ has a corresponding edge in $\mathcal{G}$, but not every edge in $\mathcal{G}$ has a corresponding edge in $\mathcal{T}$. 
Consequently, it is convenient to represent both $\mathcal{G}$ and $\mathcal{T}$ with a single data structure; however, they serve different purposes and are best thought of separately. %
Additionally, note that despite building backwards, the direction of the edges in both the tree and the graph is always defined such that applying control $e^{ij}[\control]$ from state $v^i[\state]$ yields state $v^j[\state]$.

Construction proceeds according to Alg.~\ref{alg:build}. First, a random state is sampled from the state space and the closest vertex in the graph is identified. Second, several controls are randomly sampled and applied backward in time from the chosen vertex. The control that leads to the state that is furthest from any existing vertex in the graph is chosen. We do this to prioritize even coverage of the state space but other control selection techniques, such as targeting the sampled state to prioritize certain regions of the space, are possible as well. Third, a vertex and an edge that correspond to that state and control are added to the graph and the tree. Fourth, edges from the new vertex to other existing vertices are identified. This procedure is detailed in Sec. \ref{subsec:find_connections}. Finally, the tree is updated using the added vertex and new edges in the graph, detailed in Sec. \ref{subsec:update_tree}.

\begin{algorithm}
	\caption{Build($\goalstate$,$\fdyn$,$\bdyn$,$\runningcost$) \label{alg:build}}
	    Initialize graph $\mathcal{G}$ with vertex $(\goalstate)$\\
		Initialize tree $\mathcal{T}$ with vertex $(\goalstate,0)$\\
		\While{not finished}{
			Sample state $\state_s\in\mathcal{X}$ \\
			Pick vertices $v_G,v_T$ closest to $\state_s$ \\
			Choose control $\control$ to apply from $v_G, v_T$ \\
			$\state\gets\bdyn(v_G[x],\control)$ \\
			$v_G' \gets (\state)$ \Comment{graph vertex} \\
			$e_G^{v_G' v_G}\gets(\control)$ \Comment{graph edge} \\
			Add $v_G',e_G'$ to $\mathcal{G}$ \\
			$c \gets \runningcost(\state,\control)$ \Comment{stage cost} \\
			$v_T' \gets (\state,c+v_T[J]$) \Comment{tree vertex}\\
			$e_T' \gets (\control,c)$ \Comment{tree edge} \\
			Add $v_T',e_T'$ to $\mathcal{T}$ \\
			FindConnections($v_G',\mathcal{G},\mathcal{T}$) \\
			UpdateTree($v_G',v_T',\mathcal{G},\mathcal{T}$)
		}
		\Return $\mathcal{G},\mathcal{T}$
\end{algorithm}

\subsection{Find Connections}
\label{subsec:find_connections}

This portion of the algorithm aims to identify potential parent and child vertices for a given vertex in the graph. 
Allowing a vertex to select a new parent provides an opportunity for it to improve its cost-to-go. %
Finding new potential children allows other vertices to choose it as a parent and similarly reduce their own cost-to-go.

First, controls are sampled one step forward in time from the new vertex to identify other potential parents. This works by simulating the controls from that vertex, identifying the vertices closest to the resulting states, and attempting to refine the controls so that the two coincide. 
To avoid issues related to the imprecision of gradients of (potentially neural network) dynamics, we employ a derivative-free optimization technique based upon the pattern search of \cite{hookejeeves}, in which we use several random directions instead of only the coordinate axes. For each search that reaches some threshold of the targeted vertex, an edge is added to the graph. 
Second, this process is repeated for controls going backward in time to identify potential children.

\begin{algorithm}
	\caption{FindConnections($v_G,\mathcal{G},\mathcal{T}$) \label{alg:find_connections}}
	    \For{$v_G' \in $ \text{Candidate Parents} $\subset \mathcal{V}_G$, $v_G'\neq v_G$}{
	        \If{$\exists \control \text{ s.t. } \|\fdyn(v_G[\state],\control)-v_G'[\state]\|<\varepsilon$} {
	            $e_G^{v_G v_G'} \gets (\control,\runningcost(v_G[\state],\control))$
	        }
	    }
	    \For{$v_G' \in \text{Candidate Children} \subset \mathcal{V}_G,\; v_G'\neq v_G$}{
	        \If{$\exists \control \text{ s.t. } \|\bdyn(v_G[\state],\control)-v_G'[\state]\|<\varepsilon$} {
	            $e_G^{v_G' v_G} \gets (\control,\runningcost(v_G'[\state],\control))$
	        }
	    }
\end{algorithm}

\begin{algorithm}[h]
    \caption{UpdateTree($v_G,v_T,\mathcal{G},\mathcal{T}$) \label{alg:update_tree}}
		\For{$e_G^{v_G i} \in \mathcal{E}_G$}{
		    $v_G' \gets$ graph vertex corresponding to $i$ \\
		    $v_T' \gets$ tree vertex corresponding to $v_G'$ \\
		    \If{$e_G^{v_G v_G'}[c]+v_T'[J] < v_T[J]$}{
		        Remove current edge starting at $v_T$ from $\mathcal{T}$ \\
		        $e_T^{v_T v_T'} \gets (e_G^{v_G v_G'}[\control],e_G^{v_G v_G'}[c])$ \\
		        Add $e_T^{v_T v_T'}$ to $\mathcal{T}$ \\
		        $v_T[J] \gets e_T^{v_T v_T'} + v_T'[J]$
		    }
		}
	    \For{$e_G^{i v_G} \in \mathcal{E}_G$}{
		    $v_G' \gets$ graph vertex corresponding to $i$ \\
		    $v_T' \gets$ tree vertex corresponding to $v_G'$ \\
		    \If{$e_G^{v_T' v_T}[c] + v_T[j] < v_T'[J]$}{
		        UpdateTree($v_G',v_T',\mathcal{G},\mathcal{T}$)
		    }
	    }
\end{algorithm}

\subsection{Update Tree}
\label{subsec:update_tree}

Once a vertex has determined its potential parents and children from the set of current vertices in the graph, it must update its tree information. First, it looks at its potential parents and decides whether any result in lower cost-to-go. Second, the vertex checks on its potential children and informs them of its cost-to-go. If a potential child can improve its cost-to-go by descending from the current vertex, then the \termfont{UpdateTree} routine is run on that child. 
It is important to note that since vertices may have many potential parents, the \termfont{UpdateTree} routine may be run on the same vertex multiple times. However, it will only continue with the recursion along that path through the graph if there is improvement. This prevents cycles through the graph since a path that goes through the same vertex twice cannot be an improvement, preventing infinite recursion.

\subsection{Termination Check}
\label{subsec:termination_check}

As with most graph-based planning methods \cite{lavalle1998rapidly,karaman2011sampling}, our method is \emph{anytime} and there exist a number of reasonable stopping criteria. Options include spending a fixed maximum wall-clock time, waiting until the tree contains at least a certain number or density of vertices, and halting when the tree includes a vertex near a particular location of interest. 
The corresponding planning algorithm does not fundamentally change if the stopping condition changes.

\subsection{Cost-to-Go Properties}

In this section we state and prove two fundamental properties related to the cost-to-go represented in \algname. 
In future work, we intend to extend these results to establish convergence guarantees for the value function across the continuous state space.
Let $J^*(x)$ be the optimal cost-to-go from state $x$ and $v_T[J]^k$ be the cost-to-go at iteration $k$.

\begin{proposition}[Conservative Cost-to-Go]
\label{Prop:CostToGo}
The cost-to-go at each vertex in the tree is a conservative estimate for the optimal cost-to-go, i.e., $v_T[J]^k \geq J^*(v_T[x]), \; \forall v_T\in\mathcal{T}, \; \forall k \ge 0 $.
\end{proposition}

\begin{proof}
The cost-to-go estimate at a vertex $v_T$ is computed from a path of dynamically feasible states through the graph to the goal. Therefore, the optimal cost-to-go, considering all feasible paths in the continuous space to the goal, must be less than or equal to this estimate.   
\end{proof}

\begin{proposition}[Decreasing Cost-to-Go]
The cost-to-go at every vertex in the tree is monotonically non-increasing as new vertices are added, i.e., $v_T[J]^{k+1} \leq v_T[J]^k,\; \forall v_T\in \mathcal{T}, \forall k\geq0$.
\end{proposition}
\begin{proof}
This is because (1) vertex states and graph edge controls do not change and (2) the edges in the tree are only changed when a new edge improves the cost-to-go.
\end{proof}

\begin{remark}
As is typical in sampling based planning, Prop.~\ref{Prop:CostToGo} rests on the assumption that vertices in the graph are exactly connected by dynamically feasible edges.  In implementation, connecting these edges requires a numerical optimization subject to a finite tolerance, specifically $\epsilon$ in Alg.~\ref{alg:find_connections}. With $\epsilon$ sufficiently small, we observe that Prop.~\ref{Prop:CostToGo} holds in practice despite these numerical errors.
\end{remark}

\section{Modifying the Problem at Runtime}
\label{sec:changing}

A key benefit of storing both the graph and the tree is that the tree can be recomputed online with very little effort in response to modifications to the problem parameters. This can include changing the goal state, changing the stage cost function, and/or adding new state and/or control constraints.

To do so, we first update the cost information in the edges of the graph to reflect the new problem structure.
Second, we discard the tree and replace it by finding the shortest path, as measured by number of edges traversed, through the graph to the goal from each vertex, ignoring costs. Third, we repeatedly iterate through each vertex, identifying better parents, as measured by cost-to-go, until no vertex changes in a given cycle. We must iterate in this way and not according to Alg. \ref{alg:update_tree} because here, there is no guarantee that the upstream vertices are optimal until the iterations stop.

Note that reachability will not change if only the cost function changes.
However, changing the goal or adding new constraints can make it so that some vertices no longer have a valid path through the graph to the goal. For example, if there is now a wall partitioning the state space, vertices on one side of the wall cannot reach the other side. Additionally, if an edge starts at a state or uses a control that violates a constraint, the stage cost of that edge is set to infinity. By preserving (and not pruning) vertices and edges which violate the new constraints, we retain information which may be useful if the problem is modified further.

\begin{algorithm}
	\caption{ModifyProblem \label{alg:modify}}
	\For{$e_G^{ij} \in \mathcal{E}_G$}{
	    $e_G[c] \gets \runningcost(v_G^i[\state],e_G^{ij}[\control])$
	}
	\For{$v_G \in \mathcal{V}_G$}{
	    $v_T \gets (v_G[\state],\infty)$ \\
	    $e_T \gets$ the $e_G$ that gives the shortest path 
	}
	changed $\gets$ true \\
	\While{changed}{
	    changed $\gets$ false \\
	    \For{$v_T^i,e_T^{ij}\in\mathcal{T}$}{
	        \If{$e_T^{ij}[c] + v_T^j[J] < v_T^i[J]$}{
	            $v_T^i[J] \gets e_T^{ij}[c] + v_T^j[J]$ \\
	            changed $\gets$ true
	        }
	    }
	}
\end{algorithm}

Because tree reconstruction is relatively efficient, we construct the graph only once for the least-constrained version of the problem. 
After that, constraints may be added, e.g., as information regarding obstacles is gathered from onboard sensors. %

\section{Control}
\label{sec:control}

The tree $\mathcal{T}$ serves as the basis for controlling the system since it encodes the value function for the system and control information. Each vertex contains a point in the state space and the corresponding cost-to-go, while the edges contain control actions that lead to the root of the tree (i.e., the goal). Of course, while controlling the system, the state will never perfectly match the state in any of the vertices, so some interpolation technique is required.

In \secref{sec:results}, we use a bump function, defined by
\begin{equation} \label{eq:bump}
\varphi(\state,\state') = \begin{cases}
\exp\left(-\frac{1}{1 - \left( \gamma \|\state-\state'\| \right)^2} \right), & \|\state-\state'\| < \frac{1}{\gamma} \\
0, & \text{otherwise,}
\end{cases}
\end{equation}
to interpolate the value data in the tree. To generate a control at each time step, we perform a pattern search \cite{hookejeeves} to find the control that minimizes the sum of the stage cost and the value of the subsequent state:
\begin{equation}
    \label{eqn:controller}
    u^* = \arg\min_{u\in\mathcal{U}} \runningcost(\state,\control) + \frac{\sum_{v_T\in\mathcal{T}} \varphi(\fdyn(\state,\control),v_T[\state]) v_T[J]}{\sum_{v_T\in\mathcal{T}} \varphi(\fdyn(\state,\control),v_T[\state])} .
\end{equation}

\begin{algorithm}
	\caption{Control($\state$)}
	$u \gets 0$ \Comment{Vector in $\mathbb{R}^m$} \\
	Initialize $\alpha\in\mathbb{R} > 0$ \\
	Initialize $\delta_1,..., \delta_{2m}$ with unit direction vectors in $\mathbb{R}^m$ \\
	$\state'_0 \gets \fdyn(\state,\control)$ \\
	Estimate $J(\state'_0)$ using interpolation \\
	$J_\text{best} \gets J(\state'_0)$ \\
	\While{$\alpha > \varepsilon$}{
	    $\state'_k\gets\fdyn(\state,\control+\alpha \delta_k)$ for $k\in\{1,...,2m\}$ \\
	    Estimate $J(\state'_k)$ using interpolation for each $k$ \\
	    $k_\text{min} \gets \arg\min_k J(\state'_k)$ \\
	    \eIf{$J(\state'_{k_\text{min}}) < J_\text{best}$}{
	        $J_\text{best} \gets J(\state'_{k_\text{min}})$, $\control \gets \control + \delta_{k_\text{min}}$ \\
	    }{
	        $\alpha \gets \alpha/2$ \\
	    }
	}
	\Return $\control$
\end{algorithm}

\section{Results}
\label{sec:results}

We demonstrate the capability of our approach on two systems: a two-dimensional single integrator and an inverted pendulum. These systems are informative demonstrations because state obstacles are natural to add to the former and control constraints lead to stark changes in the value function of the latter. 

\subsection{Single Integrator and State Constraints}

The single integrator has two states and two control inputs which evolve linearly as
\begin{equation}
    \dot\state = \cdyn(\state,\control) = \control,
\end{equation}
and we presume a stage cost of
\begin{equation}
    \runningcost(\state,\control) = \|\control\|_2 \Delta t.
\end{equation}
The terminal state is assumed to be $\hat x = 0$; hence, the objective is to take the shortest path to the origin. 

The initial graph construction is done in the absence of obstacles. 
The value function for this is shown in \figref{fig:singleintegratorfreevalue}. Without any additional graph building, we add a U-shaped obstacle that forces many paths to be longer. The value function after including the obstacle is shown in \figref{fig:singleintegratorobstaclevalue}.

\begin{figure}
    \centering
    \includegraphics[width=0.45\textwidth]{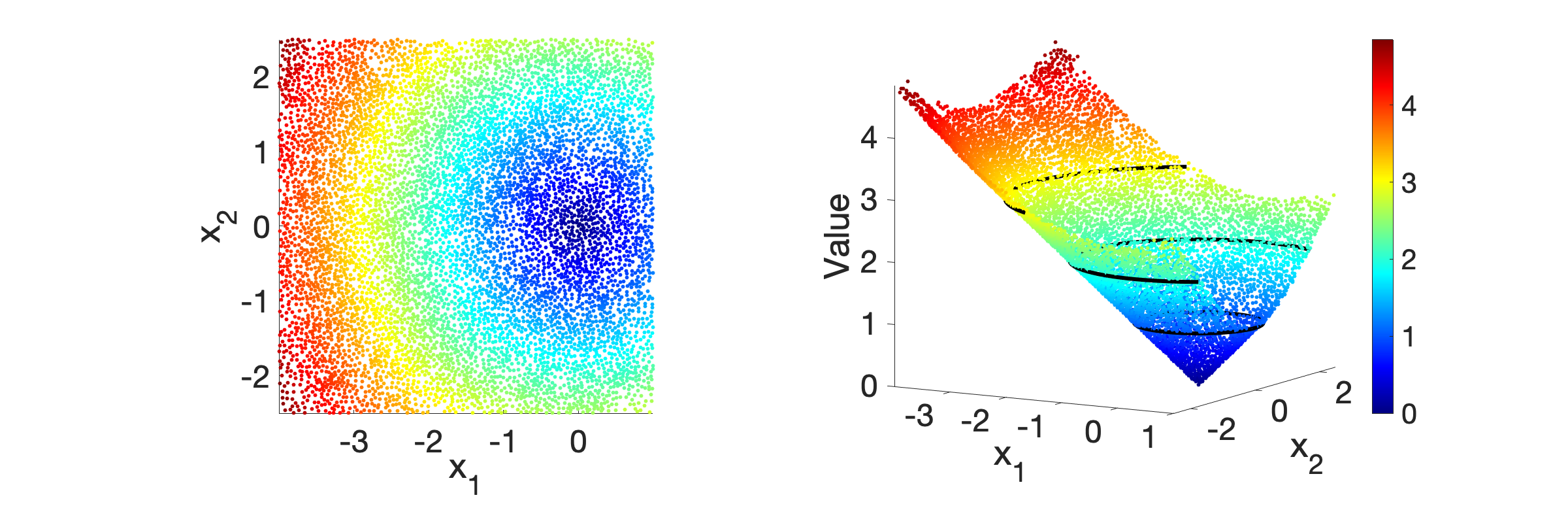}
    \caption{Values in the tree for the unconstrained single integrator system. Black rings that represent level sets for distance to the goal are added to the right plot. The tree's values match the true shortest distance to the origin. This is the result of 1 minute of computation.}
    \label{fig:singleintegratorfreevalue}
\end{figure}

\begin{figure}
    \centering
    \includegraphics[width=0.45\textwidth]{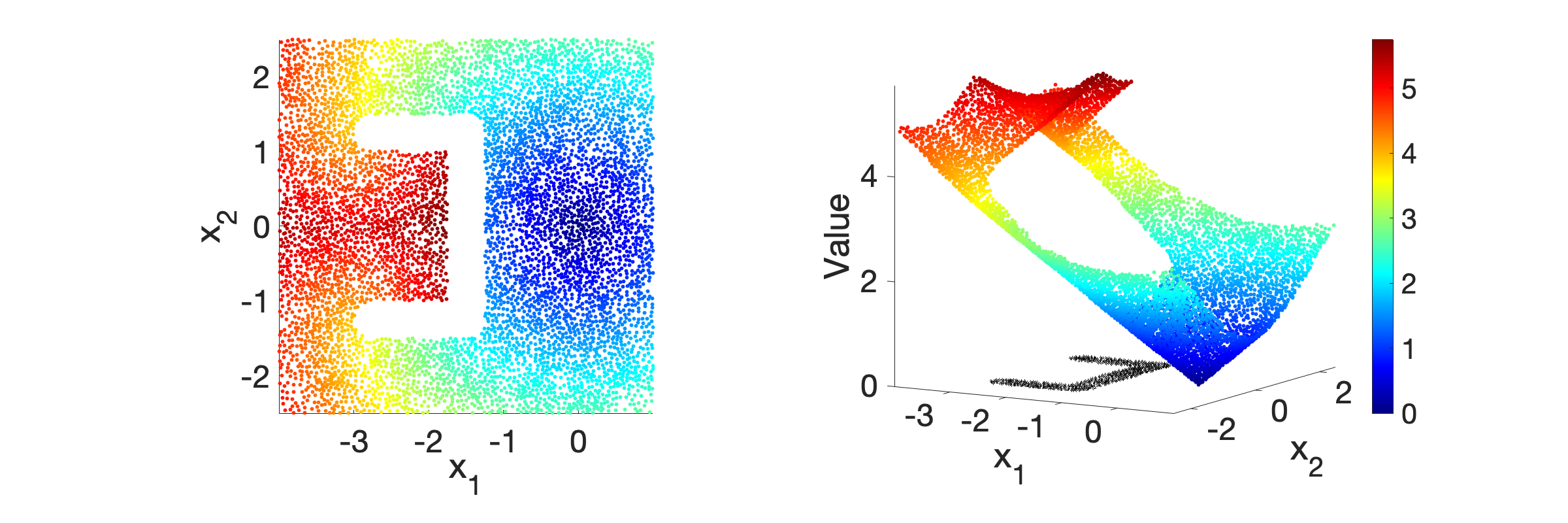}
    \caption{Values in the tree for the constrained single integrator system. The empty region on the left plot shows the added obstacle. The points within the obstacle are shown as a shadow on the right. This data was adapted from the data in Fig. \ref{fig:singleintegratorfreevalue} in 1-2 seconds.}
    \label{fig:singleintegratorobstaclevalue}
\end{figure}

\subsection{Inverted Pendulum and Control Constraints}

The inverted pendulum has two states and one control input, which follow the nonlinear model
\begin{equation} \label{eq:invpen_dynamics}
    \dot\state = F(x,u) = \begin{bmatrix}
    \state_2 \\ \sin(\state_1) + \control
    \end{bmatrix}.
\end{equation}
Stage cost is given by
\begin{equation}
    \runningcost(\state,\control) = x^\top x + u^\top u,
\end{equation}
which encodes an objective of stabilizing at the upright (unstable) equilibrium.
Given large enough control actuation limits, the inverted pendulum is feedback linearizable and hence can be treated similarly to the single-integrator. %
The value function for this setting is shown in \figref{fig:invpenoriginalvalue}. 
Indeed, it appears nearly quadratic as one would expect in an optimal control problem with linear dynamics and quadratic objective. However, when the control input is limited, the behavior is clearly nonlinear. As one expects, more restrictive control limits mean that the pendulum may need to swing back and forth to build momentum or take multiple rotations to bleed momentum. This is demonstrated in \figref{fig:invpenconstrainedvalue}. There are distinct layers and spirals in the value function, where each layer/spiral indicates another oscillation that is required. 

\begin{figure}
    \centering
    \includegraphics[width=0.45\textwidth]{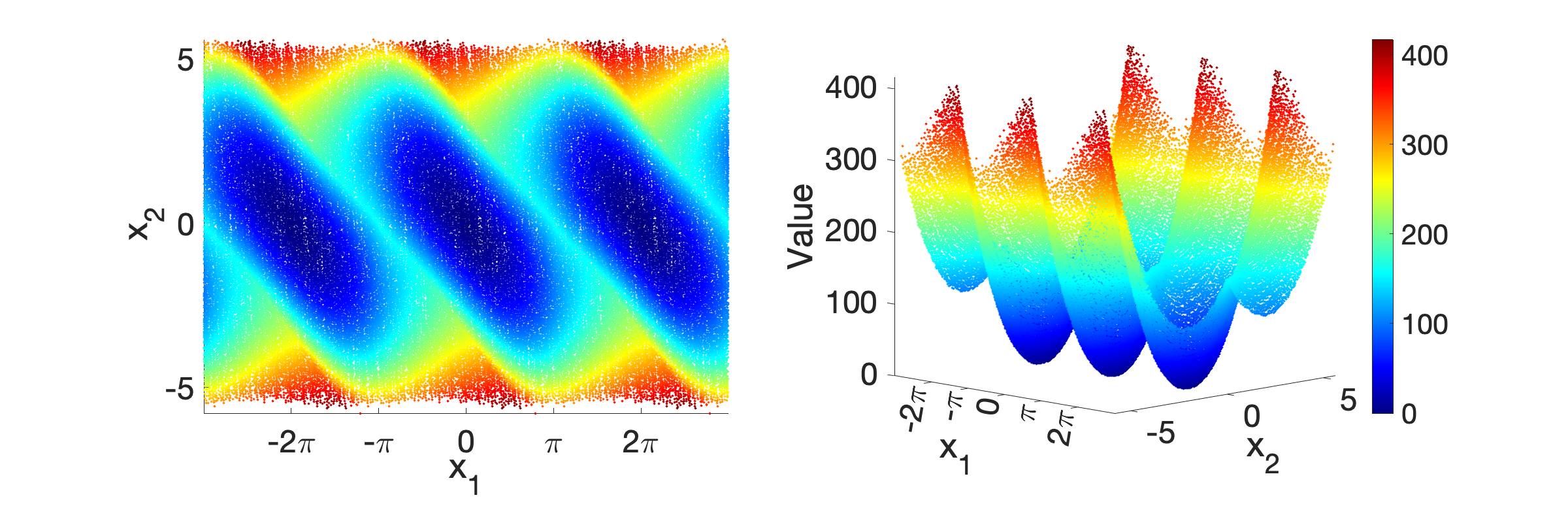}
    \caption{Values in the tree for the inverted pendulum system. The value data is repeated in the $x_1$ direction to demonstrate continuity. This is the result of 30 minutes of computation.}
    \label{fig:invpenoriginalvalue}
\end{figure}

\begin{figure}
    \centering
    \includegraphics[width=0.45\textwidth]{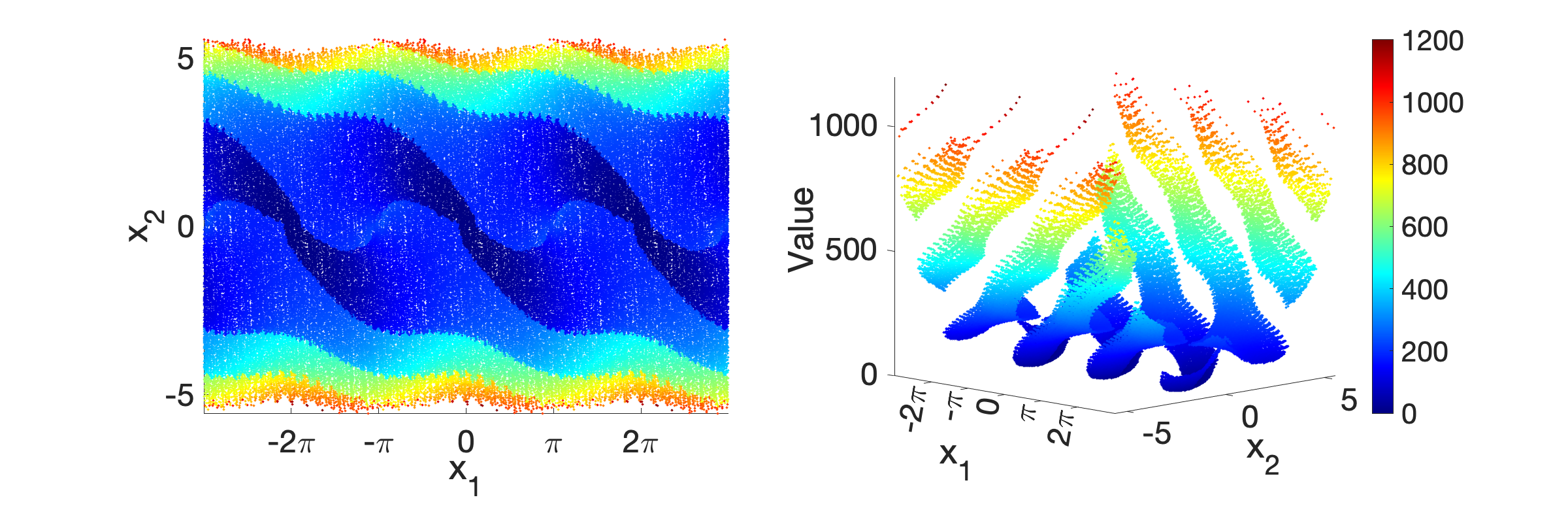}
    \caption{Values in the tree for the inverted pendulum system with more restrictive control limits. The same graph was used. Only the optimal paths through the graph changed. Again, the value data is repeated along the $x_1$ direction to demonstrate continuity. This data was adapted from the data in Fig. \ref{fig:invpenoriginalvalue} in less than 5 seconds.}
    \label{fig:invpenconstrainedvalue}
\end{figure}

\subsection{Controlling Learned Dynamics Models}

\algname~does not rely on having any knowledge of the system's physics. It only needs to be able to call functions corresponding to $\fdyn$ and $\bdyn$. This enables \algname~to control systems whose dynamics are described with a neural network, such as one learned from collected measurements of the true system.

Here, consider a system with
\begin{equation} \label{eq:network_dynamics}
    \dot\state = \cdyn(\state,\control) = \text{Network}([\state;\control]) \approx \begin{bmatrix}
    \state_2 \\ \sin(\state_1) + \control
    \end{bmatrix},
\end{equation}
where the network was trained on the inverted pendulum dynamics. However, there is inevitable mismatch between the learned model and true system, which means that following an open-loop plan is not sufficient.

We performed three simulated experiments for five different algorithms. First, we plan and execute using the physics model from \eqref{eq:invpen_dynamics}. Second, we plan and execute using the network model from \eqref{eq:network_dynamics}. Third, we plan using the network model and execute on the physics model. The five algorithms include \algname, the Open Motion Planning Library (OMPL) implementation of RRT in a single-shot approach, OMPL's RRT in an MPC-style approach, MATLAB's \termfont{fmincon} in a single-shot approach, and MATLAB's \termfont{fmincon} in an MPC approach. 

Several comments are in order. First, the control scheme for \algname~involves derivative-free optimization \eqref{eqn:controller}, in which we sample controls to see which gives the lowest expected stage cost plus cost-to-go. 
Second, the MPC version of RRT involves planning a sequence of states and controls to reach the goal and then executing that plan until the actual state differs from the planned state by more than a small threshold, at which point a new plan is computed. 
We found that this performed better than executing a fixed number of steps before replanning. 
Third, the MPC version of the \termfont{fmincon} solver adds a terminal state cost equal to 10 times the typical state cost. Without this added penalty at the end of the lookahead, the algorithm did not perform well in any case. 
Finally, at no point during the third set of experiments was any algorithm given access to the actual physics model for planning.

Figures \ref{fig:physphys}, \ref{fig:netnet}, and \ref{fig:netphys} show the state trajectories which arose during the experiments. Table \ref{tab:costs} shows the combined cost of the states and controls from each result. 

\begin{figure}
    \centering
    \vspace{0.05in}
    \includegraphics[width=0.49\textwidth]{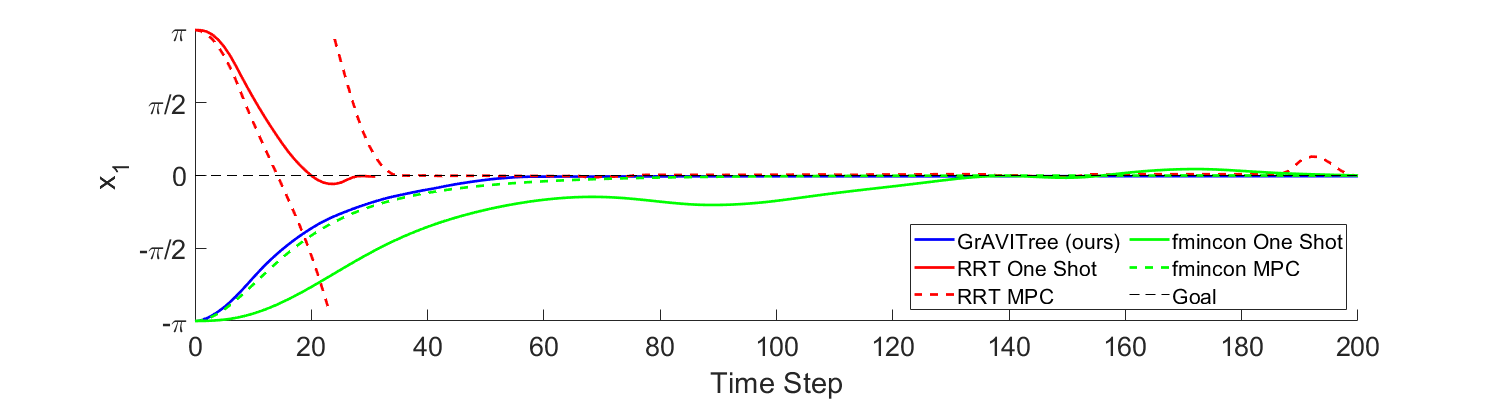}
    \caption{Planned with physics model, executed on physics model. The goal is shown as a thin, dashed black line for reference. \algname~stabilized to an error of less than 1 degree from the goal.}
    \label{fig:physphys}
\end{figure}

In \figref{fig:physphys}, we see the results from planning and executing on the physics model. All algorithms achieve some level of success in reaching the goal. \algname, one-shot RRT, and both \termfont{fmincon} controllers take fairly direct paths to the goal. However, the one-shot RRT and \termfont{fmincon} trajectories have higher costs from faster and slower motion, respectively. The slower one-shot \termfont{fmincon} trajectory may constitute a suboptimal local solution. %
The MPC RRT baseline takes an extra revolution to reach the goal and briefly falls away from the goal near the end but is otherwise successful.

\begin{figure}[tp]
    \centering
    \includegraphics[width=0.49\textwidth]{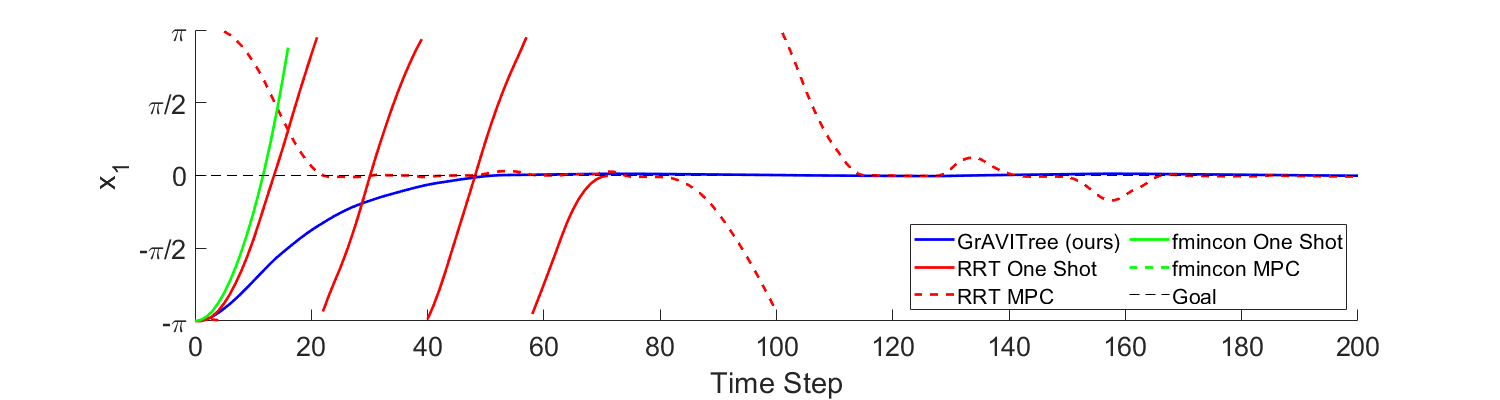}
    \caption{Planned with learned neural network model, executed on neural network model. The goal is shown as a thin, dashed black line for reference. The two \termfont{fmincon} trajectories are identical and overlap on the plot. \algname~stabilized to within approximately 2 degrees of the goal.}
    \label{fig:netnet}
\end{figure}

\figref{fig:netnet} shows the results from planning and executing on the network model. \algname~again takes a direct path to the goal, the one-shot RRT adds multiple oscillations first but still reaches the goal, and the MPC RRT is again able to reach the goal but has trouble stabilizing. The main difference here is in the \termfont{fmincon} trajectories. Although it is possible to compute a gradient of the network dynamics, curvature is often exceedingly high; hence, gradient-based algorithms such as \termfont{fmincon} struggle to cope with the resulting nonconvexities.
Both reach the plan of using maximum control effort at all times, leading state $x_1$ to perpetually increase (i.e. the pendulum spins around with increasing speed). We cut off the plot after one rotation for clarity. 

\begin{figure}
\vspace{0.05in}
    \centering
    \includegraphics[width=0.49\textwidth]{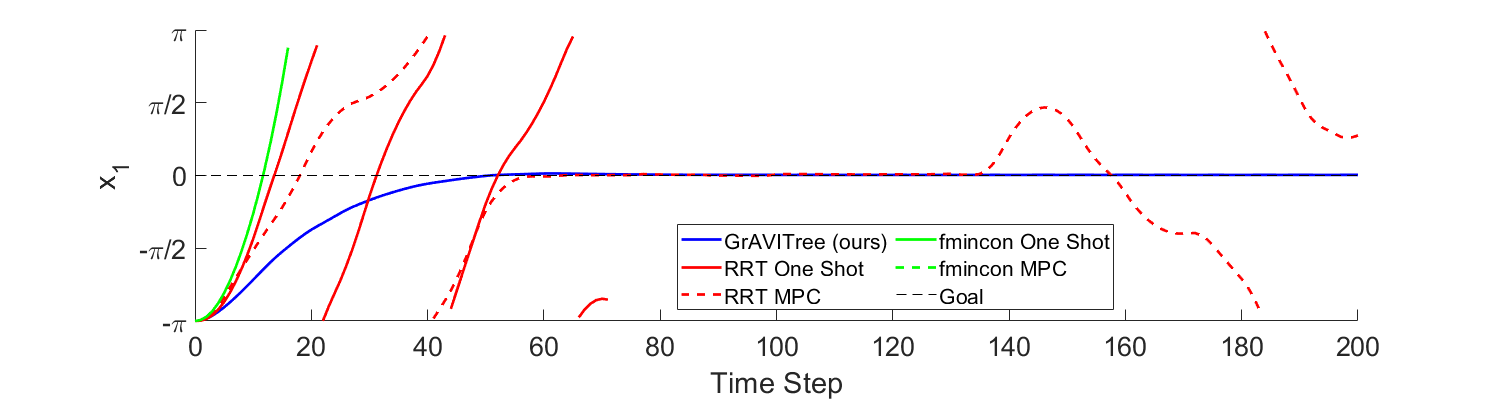}
    \caption{Planned with learned neural network model, executed on physics model. The goal is shown as a thin, dashed black line for reference. The two \termfont{fmincon} trajectories are identical and overlap on the plot. Only GrAVITree successfully stabilizes the pendulum under model mismatch. The MPC version of RRT has some success but has trouble stabilizing. \algname~stabilized to an error of less than 1 degree from the goal.}
    \label{fig:netphys}
\end{figure}

\figref{fig:netphys} shows the results from planning on the learned neural network model and executing on the physics model. This experiment includes the challenges of planning on a network model and dealing with model mismatch. \algname~again takes a direct path to the goal. The one-shot RRT, being an open loop plan, fails completely. The MPC RRT is again able to reach the goal but struggles to stabilize the system. Again, both \termfont{fmincon} approaches rapidly oscillate and the plots are cut off after one full rotation.

\begin{table}[tp]
\vspace{0.125in}
    \centering
    \resizebox{0.48\textwidth}{!}{
    \begin{tabular}{|c|c|c|c|}
    \hline
       Algorithm  & Phys-Phys & Net-Net & Net-Phys \\ \hline
        \algname & 187.15 (0.042$\pm$0.0015) & 189.65 (0.037$\pm$0.0019) & 188.15 (0.037$\pm$0.0016) \\ \hline
        RRT One & 331.28 & 1832.5 & 1652.3 \\ \hline
        RRT MPC & 1088.4 (0.021$\pm$0.033) & 1358.9 (0.031$\pm$0.049) & 1838.7 (0.075$\pm$0.014) \\ \hline
        \termfont{fmincon} One & 323.9 & Diverged & Diverged \\ \hline
        \termfont{fmincon} MPC & 187.22 (0.055$\pm$0.019) & Diverged & Diverged \\ \hline
    \end{tabular} }
    \caption{Trajectory costs. The numbers in parentheses are the mean and standard deviation of computation time for algorithms that involve recomputation. The \algname~times are dominated by the time to evaluate the value function. \algname~was run offline for 30 minutes, once for the physics model and once for the learned network model. Note that \algname~and RRT were run in C++ but \termfont{fmincon} was run in MATLAB, so the times are not directly comparable. All experiments were run on a single core of an i9-9900KS CPU.}
    \label{tab:costs}
\end{table}

\algname~is the only one of these algorithms to successfully reach and stabilize at the goal in all experiments. The MPC \termfont{fmincon} approach is very slightly worse on the first experiment, though this difference can be explained by numerical error in planning as small tweaks to parameters led to much greater differences in trajectory costs. Even when treating the MPC RRT as stabilized when first reaching the goal, the trajectory costs for \algname~were much lower. %

\section{Conclusion}
\label{sec:conclusion}

We have introduced \algname, a sampling-based algorithm for computing a value function for dynamic systems. The algorithm and its reconfiguration ability was demonstrated on two test systems: a single-integrator and an inverted pendulum. Moreover, our algorithm was used to control a system whose dynamics were modeled by a neural network and successfully handled the resulting model mismatch while remaining robust to gradient irregularities and high curvature in the neural network. Future work will aim to scale \algname~to systems with larger state spaces, more complex dynamics, and planning problems in (learned) latent spaces.

\balance
\printbibliography

\end{document}